%% file: main.tex
\documentclass[sigconf, nonacm]{acmart}

\input{preamble}
\AtBeginDocument{%
  \providecommand\BibTeX{{%
    \normalfont B\kern-0.5em{\scshape i\kern-0.25em b}\kern-0.8em\TeX}}}

\begin{document}

\title[Fast and Reliable $N - k$ Contingency Screening with Input-Convex Neural Networks]{Fast and Reliable $N - k$ Contingency Screening with \\ Input-Convex Neural Networks}

\author{Nicolas Christianson}\authornote{Part of this work was done while N. Christianson was an intern at Microsoft Research, Redmond, WA, USA.}
\affiliation{%
  \institution{California Institute of Technology}
  \city{Pasadena}
  \state{CA}
  \country{USA}
}
\email{nchristi@caltech.edu}

\author{Wenqi Cui}
\affiliation{%
  \institution{California Institute of Technology}
  \city{Pasadena}
  \state{CA}
  \country{USA}
}
\email{wenqicui@caltech.edu}

\author{Steven Low}
\affiliation{%
  \institution{California Institute of Technology}
  \city{Pasadena}
  \state{CA}
  \country{USA}
}
\email{slow@caltech.edu}

\author{Weiwei Yang}
\affiliation{%
  \institution{Microsoft Research}
  \city{Redmond}
  \state{WA}
  \country{USA}
}
\email{weiwei.yang@microsoft.com}

\author{Baosen Zhang}
\affiliation{%
  \institution{University of Washington}
  \city{Seattle}
  \state{WA}
  \country{USA}
}
\email{zhangbao@ece.uw.edu}

\renewcommand{\shortauthors}{Christianson et al.}

\begin{abstract}
  Power system operators must ensure that dispatch decisions remain feasible in case of grid outages or contingencies to prevent cascading failures and ensure reliable operation. However, checking the feasibility of all $N - k$ contingencies -- every possible simultaneous failure of $k$ grid components -- is computationally intractable for even small $k$, requiring system operators to resort to heuristic screening methods. Because of the increase in uncertainty and changes in system behaviors, heuristic lists might not include all relevant contingencies,  generating false negatives in which unsafe scenarios are misclassified as safe. In this work, we propose to use input-convex neural networks (ICNNs) for contingency screening. We show that ICNN reliability can be determined by solving a convex optimization problem, and by scaling model weights using this problem as a differentiable optimization layer during training, we can learn an ICNN classifier that is both data-driven and has provably guaranteed reliability. Namely, our method can ensure a zero false negative rate.  We empirically validate this methodology in a case study on the IEEE 39-bus test network, observing that it yields substantial ($10\text{-}20\,\times$) speedups while having excellent classification accuracy.
\end{abstract}


\maketitle

\section{Introduction}

Power systems face increasing uncertainty due to climate change, resulting from significant expansion in the development of variable renewable generation resources and environmental factors such as extreme weather events and wildfires. To ensure reliable operation in the face of this growing uncertainty, power system operators must dispatch generation resources in a manner that anticipates and is robust to potential asset outages, such as the failure of a transmission line. Failing to anticipate and prepare for such outages can lead to cascading failures that may necessitate load shedding, as occurred in the Texas blackouts in 2021 \cite{busby_cascading_2021}.

To assess and plan for the impacts of potential asset failures before they happen, system operators must perform contingency analysis to identify which failures will result in a post-failure operating point that is infeasible \cite[Chapter 3]{bienstock_electrical_2015}. In particular, NERC regulations mandate that US power systems remain stable for all $N - 1$ contingencies -- contingencies involving the loss of a single asset -- and that system operators plan for the multi-element contingencies with the most severe consequences \cite{noauthor_nerc_2010}. Assessing the security of and planning for such $N - k$ contingencies -- simultaneous losses of $k > 1$ assets -- is crucial for reliable system operation, as such correlated failures can cause severe blackouts, such as the 2003 Northeast blackout \cite{noauthor_final_2004}. However, the complexity of contingency analysis grows exponentially in the number of simultaneous failures $k$ that is considered: in a system with $N$ components, the number of such contingencies is $\Omega(N^k)$, which quickly becomes intractable for $k > 1$ in large-scale power systems.

To combat this intractability and enable the efficient screening of $N - k$ contingencies for $k > 1$, a number of approaches have been developed in the recent literature to accelerate contingency analysis. These methods fall into one of two categories: (1) heuristic approaches using, e.g., machine learning to predict contingency feasibility, and (2) exact methods that reduce computational expense by certifying feasibility of a collection of contingencies using ``representative'' constraints. However, these methods fall short on two fronts. The heuristic approaches (1) come with no rigorous guarantees on prediction accuracy or reliability; thus, they might misclassify a critical contingency as feasible, causing system outages. On the other hand, the exact methods (2), while reliable, are typically hand-designed rules which cannot take advantage of historical data to accelerate contingency analysis by focusing on the most relevant or likely contingencies for a particular power system. To enable reliable and efficient screening of higher-order $N - k$ contingencies in modern power systems, new approaches are needed to bridge the data-driven paradigm of machine learning with the strong reliability guarantees of exact methods.

\subsection{Contributions}

In this work, we confront this challenge, proposing a machine learning approach to screening $N - k$ contingencies that is fast, data-driven, and comes with \emph{provable} guarantees on reliability. In particular, we propose to use \emph{input-convex neural networks} (ICNNs) to screen arbitrary collections of contingencies for infeasibility. We define a \emph{reliable} classifier as one that never misclassifies an infeasible contingency as feasible -- that is, one that makes no false negative predictions -- and we show that ICNN reliability can be certified by solving a collection of convex optimization problems (Proposition~\ref{prop:cert}). Furthermore, we show that an unreliable ICNN can be transformed into a reliable one with zero false negative rate by suitably scaling model parameters by the solution to a convex optimization problem (Theorem~\ref{theorem:scaling}), and we propose a training methodology that enables learning over the restricted set of \emph{provably reliable ICNNs} by applying this scaling during training via a differentiable convex optimization layer (Theorem~\ref{theorem:scaling_2}, Algorithm~\ref{alg:training}). This fully differentiable approach ensures that the scaling procedure and its dependence on model parameters are accounted for during gradient descent updates; see Figure~\ref{fig:schematic} for a diagram outlining this approach.

\begin{figure}
    \centering
    \includegraphics[width=\linewidth]{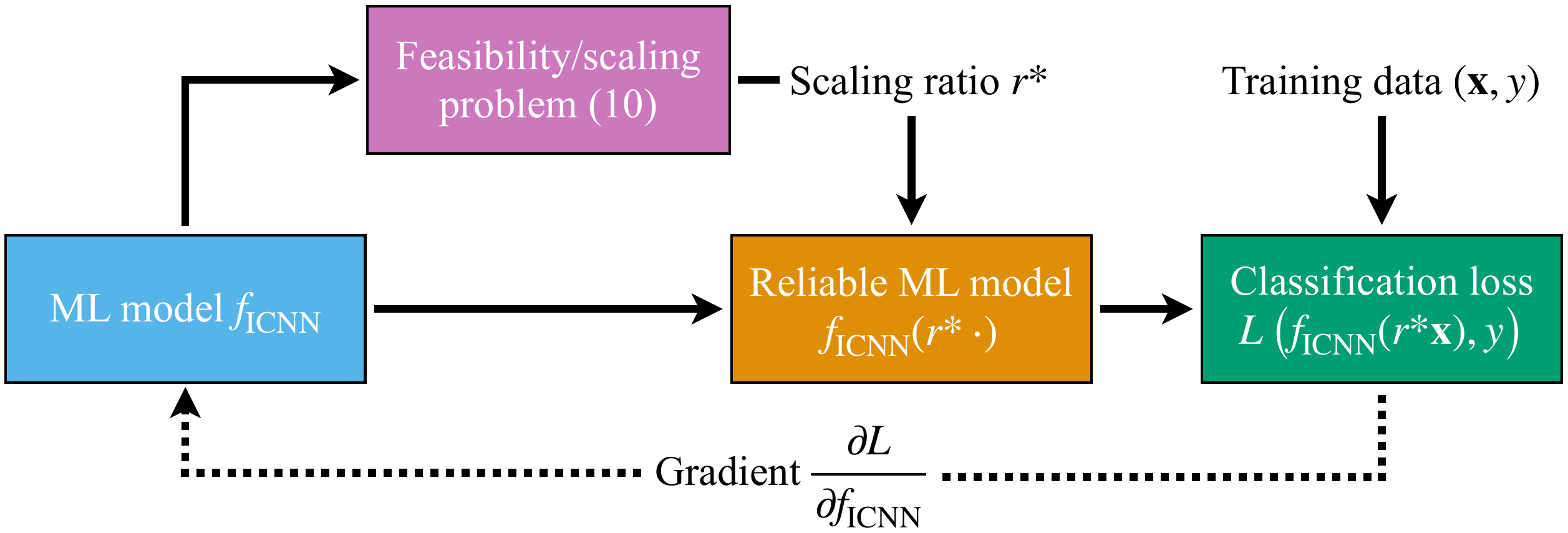}
    \caption{A schematic of our proposed methodology for training reliable classifiers for contingency screening in power systems; see Algorithm~\ref{alg:training} for a full description. Note that the scaling ratio $r^*$ is computed using a differentiable convex optimization layer, so the gradient $\partial L/\partial f_\icnn$ is aware of this scaling step.}
    \label{fig:schematic}
\end{figure}

Our approach allows for trading off the online computational burden of contingency screening for an offline one: it requires a significant computational investment during the training procedure to guarantee model reliability, but at deployment time, screening for contingencies only requires a single feedforward pass of the ICNN. We test our approach in a case study on the IEEE 39-bus test system, finding that it yields significant ($10\text{-}20\,\times$) speedups in runtime while ensuring zero false negative rate and excellent ($2\text{-}5\%$) false positive rate (Section~\ref{section:experiments_screening}). In addition, our approach yields an ICNN parametrizing an inner approximation to the set of network injections that are feasible across contingencies, which enables $10\,\times$ faster preventive dispatch via security-constrained optimal power flow (Section~\ref{section:experiments_scopf}). We anticipate that our proposed approach to learning efficient data-driven inner approximations to complex feasible sets using ICNNs could be of broader interest for other applications in energy systems and control.

\subsection{Related Work}
Our work contributes to four areas in the power systems and machine learning literature.

\noindent\textbf{Power system contingency analysis.} The problem of assessing the feasibility of contingencies has been studied in the power systems community for decades as a foundational part of secure grid operation \cite{enns_fast_1982}. Much work in recent years has sought to develop faster methods for contingency analysis, including exact methods that don't sacrifice reliability \cite{jiachun_guo_direct_2009,kaplunovich_fast_2013,kaplunovich_fast_2016} as well as heuristic and machine learning approaches that achieve faster speeds at the expense of accuracy \cite{davis_multiple_2011, crozier_data-driven_2022,nakiganda_topology-aware_2023}. Closest to our work is that of \cite{yang_fast_2017}, which proposes an approach using ``representative constraints'' to reduce the number of contingencies that must be considered; these representative constraints constitute an inner approximation of the set of all injections that are feasible across contingencies, just as our approach yields an ICNN-parametrized inner approximation to this set. In contrast to all prior approaches, our approach is both data-driven -- using ICNN models, which have substantial representational efficiency \cite[Theorem 2]{chen_optimal_2019}, to learn from system data -- and ensures rigorous guarantees on reliability, enabling fast and accurate contingency screening without any false negative predictions.

\noindent\textbf{Convex inner approximations in power systems.} The design of tractable, convex inner approximations to complicated convex or nonconvex sets is a widely studied problem in power systems and control, with applications to problems such as AC-optimal power flow \cite{dvijotham_construction_2015, nguyen_constructing_2019, lee_robust_2021} and aggregate flexibility of electric vehicles \cite{taha_efficient_2024}. When the set one wishes to approximate is convex, our approach could be adapted to enable learning such inner approximations in a data-driven manner, yielding greater efficiency and a better approximation due to the representational capacity of ICNNs.

\noindent\textbf{Machine learning in power systems.} Machine learning techniques have been applied to a wide range of problems in power systems, including contingency analysis \cite{crozier_data-driven_2022, nakiganda_topology-aware_2023}, optimal power flow \cite{zhou_deepopf-ft_2023, zhang_efficient_2023}, and security-constrained optimal power flow \cite{donti_adversarially_2021, dawson_adversarial_2023}. Of particular note in this direction are the papers \cite{chen_data-driven_2020,zhang_convex_2022,rosemberg_learning_2024}, which specifically apply \emph{ICNNs} to the problems of voltage control and optimal power flow. While some of the works applying machine learning to optimal power flow obtain generalization guarantees or provable constraint satisfaction for their methods, these guarantees hold specifically for the dispatch problem and cannot be extended to yield faster reliable approaches for contingency analysis. Thus, we give the first machine learning approach to contingency analysis with \emph{provable} guarantees on model accuracy.

\noindent\textbf{Robust and reliable machine learning.} A number of approaches have been developed to train machine learning models that are reliable in some sense, including methods to control the false positive/negative rates of a classifier \cite{elkan_foundations_2001, viaene_cost-sensitive_2005} and neural network verification and transformation techniques \cite{dvijotham_dual_2018,anderson_tightened_2020, ul_abdeen_learning_2022}. Recently, the field of \emph{learning-augmented} algorithms \cite{purohit_improving_2018,lykouris_competitive_2018} has developed new approaches to incorporate untrusted or ``black-box'' machine learning predictions into decision-making problems, including a number of energy-related problems \cite{lee_online_2021,christianson_robustifying_2022,li_learning-augmented_2023,lechowicz_online_2024}. In contrast to these prior approaches, our methodology enables learning data-driven ICNN models for contingency classification that are \emph{reliable by design}, with zero false negative rate enforced during training via a differentiable convex optimization layer.

\subsection{Notation}
Let $\R_+$ denote the nonnegative reals. Given a vector $\xvec \in \R^n$, we denote its $i$th entry $x_i$; similarly, given a matrix $\mat{M} \in \R^{m \times n}$, its $i$th row is denoted $\mvec_i$ and its $ij$th entry is denoted $m_{ij}$. Given $n \in \N$, we define $[n] = \{1, \ldots, n\}$, and given a set $\calX$, we define $\calP(\calX)$ as its power set. Given a set $\calA \subseteq \R^n$, $\inter \calA$ denotes its interior and $\vol(\calA)$ denotes its volume.

\section{Model and Preliminaries}
We begin by reviewing power network economic dispatch via the DC-optimal power flow problem and the problem of screening for infeasible contingencies. We then describe our classification approach to contingency screening and introduce the model of \emph{input-convex} neural networks we employ to this end.

\subsection{DC-OPF and Contingency Screening \label{section:model}}

Consider a power network with topology represented by a graph $G = (V, E)$, where $V$ is the set of nodes/buses and $E$ is the set of edges/transmission lines. Let $n = |V|$ be the number of buses and $m = |E|$ be the number of lines. Without loss of generality, we will assume that each bus $i \in [n]$ has a single generator.

To dispatch generation while minimizing cost and satisfying demand and other constraints in large-scale transmission networks, system operators typically solve the DC-optimal power flow (OPF) problem, which considers a linearized model of power flow \cite{stott_dc_2009}. In this problem, the system operator is faced with a known vector $\dvec \in \R^n$ of (net) demands across buses, and in response must choose generator dispatches $\pvec \in \R^n$ to minimize cost while satisfying several operational constraints:
\begin{subequations} \label{opt:economic_dispatch}
\begin{align}
    \min_{\pvec \in \R^n} &\quad \sum_{i \in [n]} c_i(p_i) \\
    \mathrm{s.t.}&\quad \underline{\pvec} \leq \pvec \leq \overline{\pvec} \label{opt:gen_limits}\\
    &\quad \bv{1}^\top (\pvec - \dvec) = 0 \label{opt:supply_demand_balance}\\
    &\quad \underline{\fvec} \leq \mat{H} (\pvec - \dvec) \leq \overline{\fvec} \label{opt:network_constraints}
\end{align}
\end{subequations}
Here, $c_i(p_i)$ is the cost for the generation decision $p_i$ on generator $i$, the constraint \eqref{opt:gen_limits} enforces lower and upper capacity limits $\underline{\pvec}, \overline{\pvec} \in \R^n$ on generation, 
\eqref{opt:supply_demand_balance} enforces supply-demand balance, and \eqref{opt:network_constraints} enforces the lower and upper bounds $\underline{\fvec}, \overline{\fvec} \in \R^m$ on line power flows given the nodal net injection vector $\pvec - \dvec$. The matrix $\mat{H}$ mapping from nodal net power injections to line power flows is specifically defined as $\mat{H} \coloneqq \mat{B}\mat{C}^\top \mat{L}^\dagger$, where $\mat{B} \in \R^{m \times m}$ is the diagonal matrix of line admittances, $\mat{C} \in \R^{n \times m}$ is a bus-by-line directed incidence matrix with entries defined as
$$c_{jl} = \begin{cases} +1 & \text{if line $l = j \to k$ for some $k \in V$}  \\ -1 & \text{if line $l = i \to j$ for some $i \in V$} \\ 0 & \text{otherwise,}\end{cases}$$
for some arbitrary orientation on the lines $E$, and $\mat{L} = \mat{C}\mat{B}\mat{C}^\top$ is the admittance-weighted network Laplacian. 

In the DC-OPF problem \eqref{opt:economic_dispatch}, the system operator solves for a feasible dispatch vector given a nominal network topology $G$. In practice, however, after a dispatch decision is chosen and the net nodal power injections $\xvec \coloneqq \pvec - \dvec$ are fixed, the network topology might change due to the failure of one or more lines. As a result of this contingency, the matrix $\mat{H}$ mapping net power injections to line power flows will change, causing the line flows to redistribute and potentially violate the line flow limits \eqref{opt:network_constraints}. Such violations may cause further lines to trip, causing a cascade of failures \cite[Chapter 4]{bienstock_electrical_2015}. Thus, to ensure continued feasible and reliable operation, the system operator must determine which contingencies are infeasible and must be planned for. This is the \emph{contingency analysis} problem, which is defined formally as follows.\footnote{Given a change in network topology resulting from a contingency, infeasibility could arise in either the line flow limits \eqref{opt:network_constraints} or the supply-demand balance constraint \eqref{opt:supply_demand_balance}; the latter is possible only in the case of \emph{islanding} contingencies which disconnect the network into multiple connected components. Because the set of islanding contingencies can be determined in advance, in this work we will restrict our focus only to the set of non-islanding contingencies and the feasibility of the line limits \eqref{opt:network_constraints}.}

\begin{problem}[\textbf{Contingency Analysis}]
    Let $\calC \subseteq \calP([m])$ be a set of contingencies of interest, where each $c \in \calC$ represents a set of failed lines, and let $\xvec = \pvec - \dvec \in \R^n$ be a vector of nodal net power injections. 
    In the \textbf{contingency analysis} problem, the system operator seeks to determine whether the net injection $\xvec$ yields feasible line flows for each contingency $c \in \calC$ -- that is, whether
    $$\underline{\fvec} \leq \mat{H}_c \xvec \leq \overline{\fvec}$$
    for each $c \in \calC$, where $\mat{H}_c = \mat{B}_c\mat{C}_c^\top \mat{L}_c^\dagger$ is defined for the post-contingency network topology with lines $E \setminus c$.
\end{problem}

A standard choice for the set of reference contingencies $\calC$ is the collection of all ${N - k}$ contingencies, or the set of all possible simultaneous failures of up to $k$ lines; in this case,
$$\calC = \{c \subseteq [m] : 1 \leq |c| \leq k\}.$$
In practice, however, it is impractical to check the feasibility of all possible ${N - k}$ contingencies for even moderately small $k$: in a network with $m$ lines, there are $\Omega(m^k)$ such possible contingencies, and so the complexity of ${N - k}$ contingency analysis grows exponentially with $k$. Instead, system operators typically only consider the $N - 1$ case, augmented with a small number of representative or problematic higher-order contingencies selected via heuristic methods. Such heuristics work well most of the time, since typically only a small number of contingencies are likely either to occur or to cause system infeasibility. However, they give no guarantees on system (in)feasibility for the broader set of possible $N - k$ contingencies for $k > 1$.

In this work, we seek to develop methods that can efficiently check whether a net injection $\xvec$ is feasible for \emph{all} contingencies in some general, large reference set $\calC$, such as the set of all $N - k$ contingencies for $k > 1$. To this end, we introduce the \emph{contingency screening} problem as a coarse-grained version of the contingency analysis problem.

\begin{problem}[\textbf{Contingency Screening}] \label{problem:batch_screening}
    Let $\calC \subseteq \calP([m])$ be a set of contingencies of interest, and let $\xvec \in \R^n$ be a vector of nodal net power injections. In the \textbf{contingency screening} problem, the system operator seeks to determine whether the net injection $\xvec$ is feasible for \textbf{all} contingencies $c \in \calC$ -- that is, whether $\xvec$ is in the \textbf{feasible region} 
    \begin{equation} \label{eq:tfr}
        \calF_\calC \coloneqq \left\{\yvec \in \R^n : \underline{\fvec} \leq \mat{H}_c \yvec \leq \overline{\fvec} \quad\forall c \in \calC\right\},
    \end{equation}
    where each $\mat{H}_c = \mat{B}_c\mat{C}_c^\top \mat{L}_c^\dagger$ is defined for the post-contingency network topology with lines $E \setminus c$. 
\end{problem}

The (true) feasible region $\calF_\calC$ defined above is the set of all net injections which remain feasible under any contingency in the set $\calC$. For the sake of notational convenience, in the rest of the paper we write this set abstractly as
\begin{equation} \label{eq:tfr_abstract}
    \calF_\calC \coloneqq \left\{\yvec \in \R^n : \mat{A} \yvec \leq \bvec\right\},
\end{equation}
where $\mat{A} \in \R^{2m|\calC| \times n}$ and $\bvec \in \R^{2m|\calC|}$ collect all of the constraints $\underline{\fvec} \leq \mat{H}_c \yvec \leq \overline{\fvec}$ in \eqref{eq:tfr}. We will assume that $\mat{A}$ contains no zero rows, since these would encode vacuous constraints. We will also make the following mild assumptions on the structure of $\calF_\calC$.

\begin{assumption} \label{assum:feasible_region}
    $\calF_\calC$ is a strict subset of $\,\R^n$ whose interior contains the origin: $\calF_\calC \subsetneq \R^n$ and $\vec{0} \in \inter \calF_\calC$. Equivalently, $\mat{A}$ has at least one row and $\mat{A}\vec{0} < \bvec$.
\end{assumption}

\noindent Note that these assumptions are reasonable ones: the first simply means that $\calF_\calC$ encodes \emph{some} constraint; if it doesn't, then we have no need to perform contingency screening. The second assumption amounts to the condition that the lower and upper line limits $\underline{\fvec}, \overline{\fvec}$ are bounded away from zero, which should hold in practice. 

The contingency \emph{screening} problem differs from the problem of contingency \emph{analysis} in that it focuses on feasibility across the entire reference set of contingencies $\calC$, rather than the feasibility of individual contingencies. We can frame this as a binary classification problem where one seeks to classify a net injection vector $\xvec \in \R^n$ as feasible or infeasible, and true labels are given by the indicator function $f_\calC$ defined as
$$f_\calC(\xvec) = \begin{cases}0 & \text{if $\xvec \in \calF_\calC$ \,(feasible)} \\ 1 & \text{otherwise (infeasible).} \end{cases}$$

While at first glance this might appear to be a simpler problem than the full contingency analysis problem, determining whether some injection $\xvec \in \calF_\calC$ (i.e., computing the label $f_\calC(\xvec)$) still has complexity $\Omega(m|\calC|)$ in general, as it requires checking the feasibility of each contingency in $\calC$. If approximations suffice, we could instead use techniques from machine learning to learn a more computationally efficient approximation of the function $f_\calC$ in a data-driven fashion using, e.g., neural networks; however, this computational speedup will typically come at the expense of reduced classification accuracy. In particular, a generic machine learning classifier might suffer \emph{false negatives}, where it classifies injections as feasible when they are not. While false positives (misclassifying a feasible injection as infeasible) may simply cause increased caution, false negatives pose a serious threat to reliable power system operation, since an infeasible injection that is not identified as such could lead to a cascade of failures. 

While the machine learning literature has developed a number of techniques to reduce the incidence of false negatives in classification, such as increasing the loss weight of examples in the positive class, none of these techniques can yield \emph{provably} guaranteed control over the false negative rate. To confidently deploy machine learning methods to contingency screening, they should ideally avoid any false negative predictions; we call such a classifier \emph{reliable}.

\begin{definition} \label{defn:reliable}
    A classifier $f : \R^n \to \{0, 1\}$ for the contingency screening problem (Problem~\ref{problem:batch_screening}) is said to be \textbf{reliable} if it has zero false negative rate, i.e., if
    $$f(\xvec) = 0 \,\text{ implies }\, \xvec \in \calF_\calC$$ 
    for any $\xvec \in \R^n$. 
\end{definition}

Note that a reliable classifier $f$ is exactly one whose \emph{predicted feasible region} $\{\xvec \in \R^n : f(\xvec) = 0\}$ is contained inside the true feasible region $\calF_\calC$; that is, the predicted feasible region should be an inner approximation of the true feasible region. Our goal in this work is to develop an approach for training reliable ML classifiers for contingency screening that satisfy this property. For general machine learning models and classification problems, determining whether this containment property holds is not typically tractable. However, as we will see in Section~\ref{section:methods}, the convex polyhedral structure of the true feasible region $\calF_\calC$ enables tractably answering this question when we restrict to a class of \emph{convex} neural networks.

\subsection{Input-Convex Neural Networks}

\emph{Input-convex neural networks (ICNNs)} \cite{amos_input_2017} are a restricted class of neural networks that parametrize convex functions. 
We consider feed-forward ICNNs $f_{\icnn} : \xvec \mapsto \yvec$ with $k$ hidden layers of the form
\begin{align*}
    \zvec_1 &= \relu\left(\mat{D}_1\xvec + \bvec_1\right) \\
    \zvec_{i} &= \relu\left(\mat{W}_{i-1}\zvec_{i - 1} + \mat{D}_i \xvec + \bvec_i\right) &\text{for $i = 2, \ldots, k$}  \tageq\label{eq:ICNN}\\
    \yvec &= \mat{W}_k\zvec_k + \mat{D}_{k+1} \xvec + \bvec_{k+1},
\end{align*}
where $\zvec_i$ is the $i$th hidden layer, the intermediate activation function is $\relu(x) = \max\{x, 0\}$, and the the weight matrices $\mat{W}_i$ are all assumed to have nonnegative entries (the weights $\mat{D}_i$ can have arbitrary entries). It is relatively straightforward to see that, under these assumptions (and more generally in the case of convex, nondecreasing activation functions), $f_\icnn(\xvec)$ is convex in $\xvec$ \cite[Proposition 1]{amos_input_2017}. Moreover, given sufficient depth and width, ICNNs can approximate \emph{any} convex function arbitrarily well \cite[Theorem 1]{chen_optimal_2019}. 

In the remainder of this work, for our application to the contingency screening problem, we will consider ICNNs with input dimension $n$ and output dimension $1$. Note that the lack of an output activation function means that the ICNN's output could be unboundedly large or small; when using an ICNN to classify the feasibility of an injection $\xvec$, we will take its prediction to be $\sigma(f_\icnn(\xvec))$, where $\sigma(x) = (1 + e^{-x})^{-1}$ is a sigmoidal activation applied to the output of the ICNN. In this case, predictions less than $0.5$ will correspond to a ``feasible'' classification (0), and those strictly greater than $0.5$ will correspond to ``infeasible'' (1). With this setup, one readily observes that the predicted feasible region of an ICNN is exactly its $0$-sublevel set:
$$\{\xvec \in \R^n : \sigma(f_\icnn(\xvec)) \leq 0.5\} = \{\xvec \in \R^n : f_\icnn(\xvec) \leq 0\}.$$
Note that the universal convex function approximation property enjoyed by ICNNs implies that \emph{any} convex set can be approximated arbitrarily well by the $0$-sublevel set of an ICNN. Thus, ICNNs are well-matched to the task of approximating the true feasible region $\calF_\calC$ for contingency screening, which is itself a convex set.

Following Definition~\ref{defn:reliable}, a reliable ICNN classifier is one whose predicted feasible region is contained inside the true feasible region $\calF_\calC$. In the next section, we will discuss how the convex structure of an ICNN enables both (a) tractably determining whether this containment property holds and (b) scaling an ICNN's parameters to guarantee its reliability.

\section{Certifying and Enforcing Reliability for ICNN Contingency Classifiers \label{section:methods}}

As discussed in Section~\ref{section:model}, a \emph{reliable} classifier for the contingency screening problem is one that makes no false negative predictions, i.e., whose predicted feasible region is fully contained inside the true feasible region $\calF_\calC$ \eqref{eq:tfr}. For an ICNN classifier $f_\icnn$, this amounts to the property that its $0$-sublevel set is contained in $\calF_\calC$. An immediate question that arises is whether it is possible to certify that a given classifier $f_\icnn$ satisfies this reliability criterion. Conveniently, we can show that certifying this property reduces to solving a collection of convex optimization problems. 

\begin{proposition} \label{prop:cert}
    An ICNN classifier for the contingency screening problem is reliable -- i.e., has zero false negative rate -- if and only if
    \begin{equation} \label{eq:cert_opt}
        \left\{\begin{aligned}
            \max_{\xvec \in \R^n} &\quad \avec_j^\top \xvec \\
            \mathrm{s.t.}&\quad f_\icnn(\xvec) \leq 0
        \end{aligned}\right\} \leq b_j
    \end{equation}
    for all $j \in \left[2m|\calC|\right]$, where $\avec_j$ is the $j$th row of $\mat{A}$.
\end{proposition}
\begin{proof}
    We first observe that, since $f_\icnn$ is a convex function, the optimization problem in \eqref{eq:cert_opt} is a convex problem, and thus can be solved tractably. Given this convexity, the fact that containment of the $0$-sublevel set of $f_\icnn$ inside the polyhedron $\calF_\calC$ can be determined by solving a collection of convex optimization problems of the form \eqref{eq:cert_opt} is a standard result in convex optimization (see, e.g., \cite{eaves_optimal_1982}). For the sake of completeness, we briefly describe the proof here.

    For the forward direction, suppose that containment holds, i.e., $\{\xvec \in \R^n : f_\icnn(\xvec) \leq 0\} \subseteq \calF_\calC$. This means that $f_\icnn(\xvec) \leq 0$ implies $\mat{A}\xvec \leq \bvec$; thus any feasible solution $\xvec$ to the problem in \eqref{eq:cert_opt} will satisfy the inequality $\avec_j^\top \xvec \leq b_j$, and hence this inequality will hold at optimality.

    For the reverse direction, suppose that \eqref{eq:cert_opt} holds for all $j$. If there were some $\xvec \in \R^n$ which was not feasible ($\xvec \not\in \calF_\calC$) and yet was predicted feasible by the ICNN ($f_\icnn(\xvec) \leq 0$), this would imply the existence of some $j$ such that $\avec_j^\top \xvec > b_j$, yielding a contradiction.
\end{proof}

The previous proposition provides a way of tractably certifying whether a given ICNN classifier is reliable, but it does not give a means of transforming an unreliable classifier into a reliable one. Since reliability of a classifier is exactly containment of its predicted feasible set inside the true feasible set, a natural approach to enforcing reliability would be to transform the classifier to translate and scale its predicted feasible set into the interior of $\calF_\calC$. In general, the problem of scaling some convex set $\calA$ to be contained in another convex set $\calB$ can be tractably cast as a convex optimization problem in certain special cases, such as when both $\calA$ and $\calB$ are polyhedra in halfspace form (see the foundational work of Eaves and Freund~\cite{eaves_optimal_1982}). However, the set we are concerned with scaling is the $0$-sublevel set of an ICNN, which has not been considered in prior work and is considerably more complex given the multilayer nature and substantial representational efficiency of ICNNs \cite[Theorem 2]{chen_optimal_2019}.

Nonetheless, as we show in the following theorem, it is possible to perform such a scaling efficiently by solving a collection of convex optimization problems, yielding a reliable classifier.

\begin{theorem} \label{theorem:scaling}
    Let $r^*$ and $\vvec^*$ be the optimal solutions to the optimization problem 
    \begin{subequations} \label{opt:scaling}
        \begin{align}
            \min_{r \in \R_+, \vvec \in \R^n} &\quad r \label{opt:scaling_obj} \\
            \mathrm{s.t.} &\quad z_j^* \leq \avec_j^\top \vvec + b_j r \quad \forall j \in \left[2m|\calC|\right], \label{opt:scaling_constr}
        \end{align}
    \end{subequations}
    where
    \begin{equation} \label{opt:z_j}
        \begin{aligned}
            z_j^* \coloneqq \max_{\xvec \in \R^n}&\quad \avec_j^\top \xvec \\
            \mathrm{s.t.}&\quad f_\icnn(\xvec) \leq 0
        \end{aligned}
    \end{equation}
    for each $j \in \left[2m|\calC|\right]$. Then the transformed ICNN classifier $\hat{f}_\icnn$ defined as
    $$\hat{f}_\icnn(\xvec) \coloneqq f_\icnn(r^*\xvec + \vvec^*)$$
    has zero false negative rate. Moreover, \eqref{opt:scaling} has a feasible solution as long as 
    the original predicted feasible set $\{\xvec \in \R^n : f_\icnn(\xvec) \leq 0\}$ is bounded.
\end{theorem}

Before proving Theorem~\ref{theorem:scaling}, we first make four brief comments. First, note that the boundedness assumption on the predicted feasible set $\{\xvec \in \R^n : f_\icnn(\xvec) \leq 0\}$ can be easily enforced by, e.g., adding an indicator function to $f_\icnn$ that is $0$ for all $\|\xvec\| \leq D$ and $+\infty$ otherwise, where $D$ is some large constant.

Second, note that the transformed classifier $\hat{f}_\icnn$ can be obtained from $f_\icnn$ (as defined in \eqref{eq:ICNN}) by multiplying the weights $\mat{D}_i$ by $r^*$ and adding $\mat{D}_i\vvec^*$ to the biases. Its predicted feasible set is a transformed version of $f_\icnn$'s, obtained by translating by $-\vvec^*$ and scaling down by a factor of $r^*$. As long as $r^*$ is not infinite -- that is, if \eqref{opt:scaling} is feasible -- then the predicted feasible set of $\hat{f}_\icnn$ will be nonempty (assuming that of $f_\icnn$ is nonempty). We thus seek to minimize $r$ so as to maximize the volume of $\hat{f}_\icnn$'s predicted feasible set, which will ensure reliability while minimizing the conservativeness of this classifier as an inner approximation of the true feasible set $\calF_\calC$. Note that the resulting classifier might still be relatively conservative and suffer poor prediction accuracy on the negative class, i.e., a large false positive rate; in Section~\ref{section:dcl} we will propose a methodology to reduce this conservativeness and enforce classifier reliability \emph{during training} by incorporating a version of the scaling problem \eqref{opt:scaling} into the training process as a differentiable layer.

Third, observe that computing $r^*$ and $\vvec^*$ requires solving a collection of $2m|\calC|$ optimization problems \eqref{opt:z_j} followed by a linear program \eqref{opt:scaling} with just as many constraints. One might question, thus, the benefit of our scaling approach over exhaustive checking of contingencies, which has a similar dependence on $|\calC|$ in its complexity. However, our approach has a substantial benefit: this scaling must only be performed \emph{once} to obtain a classifier that is provably reliable for \emph{any} net injection, and all subsequent feasibility predictions only require an efficient feedforward pass of the ICNN. In contrast, exhaustively checking contingencies must be done separately for \emph{every} net injection. Thus, our approach yields significantly improved efficiency at deployment time by moving the computational burden of ensuring reliability from the \emph{online}, real-time setting to an \emph{offline} preprocessing step. 

Finally, we note that it is possible to transform the problems \eqref{opt:scaling}, \eqref{opt:z_j} into a single linear program by taking the Lagrange dual of each maximization problem \eqref{opt:z_j} \cite{yeh_end--end_2024}, similar to the approach for polyhedra in \cite{eaves_optimal_1982}. However, our multi-problem formulation is more efficient, as it lends itself to a distributed solution approach where we solve each of the smaller, independent optimization problems \eqref{opt:z_j} in parallel before using their optimal solutions to solve the linear program \eqref{opt:scaling}.

We now present a proof of Theorem~\ref{theorem:scaling}.

\begin{proof}[Proof of Theorem~\ref{theorem:scaling}]
    Consider the optimization problem
    \begin{subequations} \label{opt:vol_scale_form}
    \begin{align}
        \max_{r \in \R_+, \vvec \in \R^n} &\,\, \vol\left(\{\xvec \in \R^n : f_\icnn(r\xvec + \vvec) \leq 0\}\right) \label{opt:vol_scale_form_obj}\\
        \mathrm{s.t.} &\,\, \left\{\begin{aligned}
            \max_{\xvec \in \R^n} &\,\, \avec_j^\top \xvec \\
            \mathrm{s.t.}&\,\, f_\icnn(r\xvec + \vvec) \leq 0
        \end{aligned}\right\} \leq b_j \quad \forall j \in \left[2m|\calC|\right] \label{opt:vol_scale_form_constr}
    \end{align}
    \end{subequations}
    where we seek to maximize the volume of the predicted feasible set of $f_\icnn$ (to minimize conservativeness) after scaling and translating it by $r$ and $\vvec$, subject to the constraint that this transformed set is contained in the true feasible set $\calF_\calC$. First, note that since $\calF_\calC$ has nonempty interior -- and specifically, $\vec{0} \in \inter \calF_\calC$ (Assumption~\ref{assum:feasible_region}) -- then if the original predicted feasible set $\{\xvec \in \R^n : f_\icnn(\xvec) \leq 0\}$ is bounded, then \eqref{opt:vol_scale_form} has a feasible solution. This is because there must be a $\varepsilon$-neighborhood about the origin that remains contained in $\calF_\calC$; thus, since the predicted feasible region of $f_\icnn$ is bounded, it is possible to choose a translation $\vvec$ and a sufficiently large (yet finite) $r$ to ensure the transformed predicted region is contained in this $\varepsilon$-neighborhood.

    Now, let us consider the objective \eqref{opt:vol_scale_form_obj} and the constraints \eqref{opt:vol_scale_form_constr} separately. We can assume that $r > 0$, since $r = 0$ would only be feasible if $\calF_\calC$ were all of $\R^n$, which violates Assumption~\ref{assum:feasible_region}. For the objective, observe that
    \begin{align*}
        &\vol\left(\{\xvec \in \R^n : f_\icnn(r\xvec + \vvec) \leq 0\}\right) \\
        =\ &\vol\left(\{r^{-1}(\yvec - \vvec) \in \R^n : f_\icnn(\yvec) \leq 0, \yvec \in \R^n\}\right) \\
        =\ &r^{-n}\cdot \vol\left(\{\yvec \in \R^n : f_\icnn(\yvec) \leq 0\}\right), \tageq\label{eq:reduced_volume_form}
    \end{align*}
    where the final equality follows from the fact that homogeneously scaling a body by $s$ in $n$ dimensions scales the volume by $s^n$, and translation has no impact on volume. Since the volume term in \eqref{eq:reduced_volume_form} is independent of the decision variables $r$ and $\vvec$, and maximizing $r^{-n}$ will yield the same optimal solution as minimizing $r$ (since the function $s \mapsto s^{-1/n}$ is strictly decreasing on $s > 0$), we can replace \eqref{opt:vol_scale_form_obj} with $\min_{r \in \R_+, \vvec \in \R^n} \, r$ while keeping the same optimal solution. This exactly matches the objective in \eqref{opt:scaling_obj}.

    Next, consider the constraints \eqref{opt:vol_scale_form_constr}. By Proposition~\ref{prop:cert}, these constraints enforce the reliability -- or zero false negative rate -- of the transformed classifier $f_\icnn(r\xvec + \vvec)$. For a given $j \in \left[2m|\calC|\right]$, since $r > 0$, we have
    \begin{align*}
        &\left\{\begin{aligned}
            \max_{\xvec \in \R^n} &\,\, \avec_j^\top \xvec \\
            \mathrm{s.t.}&\,\, f_\icnn(r\xvec + \vvec) \leq 0
        \end{aligned}\right\} \leq b_j \\
        \iff&\left\{\begin{aligned}
            \max_{\yvec \in \R^n} &\,\, \avec_j^\top r^{-1}(\yvec - \vvec) \\
            \mathrm{s.t.}&\,\, f_\icnn(\yvec) \leq 0
        \end{aligned}\right\} \leq b_j \\
        \iff&\left\{\begin{aligned}
            \max_{\yvec \in \R^n} &\,\, \avec_j^\top \yvec \\
            \mathrm{s.t.}&\,\, f_\icnn(\yvec) \leq 0
        \end{aligned}\right\} \leq \avec_j^\top\vvec + b_j r
    \end{align*}
    which exactly matches \eqref{opt:scaling_constr} and \eqref{opt:z_j}. 
\end{proof}

\section{Training Reliable ICNN Classifiers with Differentiable Convex Optimization Layers \label{section:dcl}}

Theorem~\ref{theorem:scaling} in the previous section provides an approach to scale the parameters of an existing ICNN classifier to guarantee provable reliability, or zero false negative rate. However, this post-hoc scaling process could yield significant conservativeness. This is because scaling down the predicted feasible region by a factor of $r > 1$ decreases its volume by a factor of $r^n$; under mild assumptions on the probability distribution over net injections $\xvec \in \R^n$ seen at deployment time, this scaling could beget an exponential increase in the false positive rate compared to the original, unreliable classifier. 

To avoid this conservativeness, it is necessary to incorporate this scaling procedure into the training of the ICNN classifier, rather than applying it only after training. A natural approach is as follows: at each epoch of training, first solve the problems \eqref{opt:scaling} and \eqref{opt:z_j} to determine the optimal scaling parameters $r^*$ and $\vvec^*$. Then, evaluate the training loss of the transformed ICNN classifier -- for a single injection/label pair $(\xvec, y)$, we denote this loss $L\left(f_\icnn(r^*\xvec + \vvec^*), y\right)$, where $L$ is some classification loss -- and update the model $f_\icnn$ using the gradient $\frac{\partial L}{\partial f_\icnn}$, where $\partial f_\icnn$ refers to the gradient with respect to all the parameters of $f_\icnn$. This approach aligns the training loss with the objective of learning the optimal reliable classifier, since the loss that is minimized through gradient descent is that of the reliable, scaled ``version'' of the generic classifier $f_\icnn$.

As currently described, however, this approach is incomplete. In particular, note that the scaling parameters $r^*, \vvec^*$ resulting from the problem \eqref{opt:scaling} themselves depend on the parameters of $f_\icnn$ through each $z_j^*$. Defining $\hat{y} \coloneqq f_\icnn(r^*\xvec + \vvec^*)$, by the chain rule, the gradient of $L(\hat{y}, y)$ with respect to the parameters of $f_\icnn$ is
\begin{align*}
    &\frac{\partial L}{\partial f_\icnn}(\hat{y}, y) \\
    &=\frac{\partial L}{\partial \hat{y}} \left(\frac{\partial \hat{y}}{\partial f_\icnn} + \frac{\partial \hat{y}}{\partial r^*}\sum_j\frac{\partial r^*}{\partial z_j^*} \frac{\partial z_j^*}{\partial f_\icnn} + \frac{\partial \hat{y}}{\partial \vvec^*}\sum_j\frac{\partial \vvec^*}{\partial z_j^*} \frac{\partial z_j^*}{\partial f_\icnn}\right).
\end{align*}
Thus to compute the gradient of the loss $L$ with respect to the parameters of the ICNN $f_\icnn$, it is necessary to also compute the gradients of the optimal solutions $r^*, \vvec^*$ of \eqref{opt:scaling} with respect to each $z_j^*$, and the gradient of each optimal value $z_j^*$ of \eqref{opt:z_j} with respect to $f_\icnn$'s parameters. To compute these gradients, we can employ differentiable convex optimization layers \cite{agrawal_differentiable_2019}, which automatically compute the gradient of a convex optimization problem with respect to problem parameters by differentiating through the Karush-Kuhn-Tucker (KKT) conditions of the problem, allowing the incorporation of such problems into machine learning training methodologies in a fully differentiable manner. By computing $r^*$ and $\vvec^*$ using differentiable layers, we ensure that the training process is ``aware'' of the scaling procedure that is applied to $f_\icnn$ to guarantee reliability.

While this fully differentiable approach ensures that the scaling procedure is accounted for when computing the loss gradient, it requires computing both the solution to \eqref{opt:scaling} and the solutions to \eqref{opt:z_j} for all $j \in \left[2m|\calC|\right]$ using differentiable layers, which typically require additional computational overhead beyond solving the relevant optimization problems in a non-differentiable manner \cite{agrawal_differentiable_2019}. Because we need to apply this scaling at each epoch of training to enforce reliability, reducing the number of differentiable optimization layers used at each step of training would improve computational efficiency.

Fortunately, as we show in the following theorem, it is possible to obtain a fully differentiable scaling procedure using just a \emph{single} differentiable optimization step.

\begin{theorem} \label{theorem:scaling_2}
    Let $z_j^*$ be defined as in \eqref{opt:z_j} for each $j \in \left[2m|\calC|\right]$, and let $j^* \coloneqq \argmax_j z_j^*/b_j$. Define $r^*$ to be the optimal value of the following problem:
    \begin{equation} \label{opt:scaling_2}
        \begin{aligned}
            r^* \coloneqq \max_{\xvec \in \R^n} &\quad \avec_{j^*}^\top \xvec / b_{j^*} \\
            \mathrm{s.t.}&\quad f_\icnn(\xvec) \leq 0.
        \end{aligned}
    \end{equation}
    Then the transformed ICNN classifier $\hat{f}_\icnn$ defined as
    $$\hat{f}_\icnn(\xvec) \coloneqq f_\icnn(r^*\xvec)$$
    has zero false negative rate. Moreover, \eqref{opt:scaling_2} has a feasible solution as long as the original predicted feasible set $\{\xvec \in \R^n : f_\icnn(\xvec) \leq 0\}$ is bounded.
\end{theorem}
\begin{proof}
    Consider the optimization problem \eqref{opt:scaling}, and fix $\vvec = \vec{0}$; this problem remains feasible, by the assumption that the predicted feasible set is bounded, and since $\vec{0} \in \inter \calF_\calC$ (Assumption~\ref{assum:feasible_region}) implies that $\bvec > \vec{0}$. The optimal solution $r^*$ to \eqref{opt:scaling} is the smallest value of $r$ that still satisfies the constraints \eqref{opt:scaling_constr}; this is exactly
    $$r^* \coloneqq \max_j z_j^* / b_j.$$
    It is straightforward to see that this $r^*$ is identical to the one obtained by \eqref{opt:scaling_2}. Thus, the scaling obtained from \eqref{opt:scaling_2} inherits the zero false negative rate property of \eqref{opt:scaling}.
\end{proof}

In Theorem~\ref{theorem:scaling_2}, the values $z_j^*$ only need to be computed in order to determine the maximizing index $j^*$; then, the scaling ratio $r^*$ is computed using just the single optimization problem \eqref{opt:scaling_2}. As such, all of the $z_j^*$ can be computed in a non-differentiable fashion, and only \eqref{opt:scaling_2} must be solved using a differentiable layer. Note additionally that the lack of a translation variable $\vvec$ in \eqref{opt:scaling_2} shouldn't yield any additional conservativeness during training, since during training the ICNN can learn biases that would imitate the impact of any such possible $\vvec$. 

We outline in Algorithm~\ref{alg:training} a training methodology incorporating the fast, differentiable scaling procedure in Theorem~\ref{theorem:scaling_2}. In this process, we begin by ``warm-starting'' the training for $M_w$ epochs by performing standard gradient descent on the classification loss without scaling for reliability. Then, for each of the remaining $M_s$ epochs, the model is scaled using a differentiable layer implementing \eqref{opt:scaling_2} before evaluating the training loss. Note that after every gradient step, the ICNN's weights $\mat{W}_i$ must be clipped to the positive orthant to maintain convexity.

\begin{algorithm}[t]
\SetAlgoLined
\caption{Training procedure for reliable ICNN classifiers}
\label{alg:training}

\KwInput{training data $\{(\xvec_i, y_i)\}_{i=1}^N$, initial ICNN $f_\icnn$, warm-start epochs $M_w$, scaling epochs $M_s$, batch size $s$}

\BlankLine
\tcc{Warm-start the ICNN training without scaling}
\SetKwFor{ForEach}{for each}{do}{end for}
\ForEach{epoch in $[M_w]$}{
    \ForEach{mini-batch $B \subset [N]$}{
        Evaluate the loss $\frac{1}{s}\sum_{i \in B} L\left(f_\icnn(\xvec_i), y_i\right)$ \\
        Compute the gradient $\frac{\partial \,\mathrm{loss}}{\partial f_\icnn}$ and use it to update $f_\icnn$
    }
}

\BlankLine
\tcc{Train with scaling to enforce reliability}
\ForEach{epoch in $[M_s]$}{
    Compute
    \begin{equation*}
    \begin{aligned}
        \hspace{-10em}z_j^* \coloneqq \max_{\xvec \in \R^n}&\quad \avec_j^\top \xvec \\
        \mathrm{s.t.}&\quad f_\icnn(\xvec) \leq 0
    \end{aligned}
    \end{equation*}
    for each $j \in \left[2m|\calC|\right]$ \\
    \nl Set $j^* \coloneqq \argmax_j z_j^* / b_j$ \\
    Compute 
    \begin{equation*}
    \begin{aligned}
        \hspace{-10em}r^* \coloneqq \max_{\xvec \in \R^n}&\quad \avec_{j^*}^\top \xvec / b_{j^*} \\
        \mathrm{s.t.}&\quad f_\icnn(\xvec) \leq 0
    \end{aligned}
    \end{equation*}
    using a differentiable convex optimization layer \label{algline:diff_opt}\\
    Evaluate the loss $\frac{1}{s}\sum_{i \in B} L\left(f_\icnn(r^*\xvec_i), y_i\right)$ of the scaled model on a mini-batch $B$ \\
    Compute the gradient $\frac{\partial \,\mathrm{loss}}{\partial f_\icnn}$ and use it to update $f_\icnn$
}
\end{algorithm}

\section{Experimental Results}
In this section, we describe the results of our ICNN training methodology (Algorithm~\ref{alg:training}) in a case study of $N - 2$ contingency screening on the IEEE 39-bus test network \cite{athay_practical_1979, pai_energy_1989}. All experiments were performed on a MacBook Pro with 12-core M3 Pro processor, and the code for implementing the experiments is available upon request. 

We used the IEEE 39-bus test network implemented in pandapower \cite{thurner_pandapoweropen-source_2018}. We generated 14,000 random demand vectors from a multivariate normal distribution centered at the nominal demand with relative standard deviation $15\%$ and random covariance. We assigned each generator a linear cost with random coefficient between 10 and 50, and set line limits uniformly to 1600 MW. We then solved the DC-OPF problem \eqref{opt:economic_dispatch} for each demand instance to obtain net injections, which were then standardized and split into a 10,000 sample training set, a 2,000 sample validation set, and a 2,000 sample test set. 

To construct the true feasible set $\calF_\calC$, we took the set of all $N - 2$ contingencies and dropped any islanding contingencies as well as contingencies that were infeasible more than 90\% of the time, since these should be handled separately. We eliminated any dimensions for which the generated injection data was constant and eliminated redundant constraints using the method from \cite[Theorem 2]{yang_fast_2017}, using as a bounding box the empirical dimension-wise minimum and maximum net injections, multiplied by 1.2 for buffer and extended to include the origin. This resulted in a constraint matrix $\mat{A}$ with 3,613 rows and 26 columns. To account for the standardized training data, we multiplied the rows of $\mat{A}$ by $\bm{\sigma}$ and subtracted $\mat{A}\bm{\mu}$ from $\bvec$, where $\bm{\mu}$ and $\bm{\sigma}$ are the dimension-wise mean and standard deviation of the unstandardized training data.

We trained both ICNNs and standard, nonconvex neural networks (NNs) for the contingency screening task using PyTorch \cite{paszke_pytorch_2019}. All networks had a hidden width of 50, we enforced boundedness of the predicted feasible set by adding a layer ensuring the output would always be positive outside of the aforementioned bounding box of net injections, and we trained models using hidden depths of 1, 2, and 3, as well as weights of $0.5$, $1$, and $1.5$ on the positive class of the binary cross-entropy loss to probe the impact of positive class weight on false negative rate. For each choice of parameters, we trained 3 models with independent seeds, and in our results we report the mean and standard deviation of performance over these seeds. We trained the ICNNs using 500 warm-start epochs and 9,500 scaling epochs, and the nonconvex NNs were trained using 10,000 standard epochs. The cvxpylayers library \cite{agrawal_differentiable_2019} was used to differentiably solve the optimization problem in line~\ref{algline:diff_opt} of the training methodology (Algorithm~\ref{alg:training}).
We used the Adam optimizer \cite{kingma_adam_2017} with learning rate $10^{-2}$, decreasing the learning rate by a factor of 10 at epoch 1,500 and again at 8,500. During each training run, we kept track of the false positive rate on the validation set at each epoch and selected as the training output the model with the best such validation set performance.

We show in Figure~\ref{fig:slice} a 2-dimensional slice of the true feasible region $\calF_\calC$ and the predicted feasible region of a 1-layer ICNN trained via our methodology. It is evident that the ICNN respects the inner approximation property as a result of the scaling procedure while learning to focus on the data-intensive region at the bottom of the true feasible region. The ICNN does not need to learn the shape of the entire true feasible region due to data sparsity at the top of this slice, enabling a more efficient representation. 

\begin{figure}[t]
    \centering
    \includegraphics{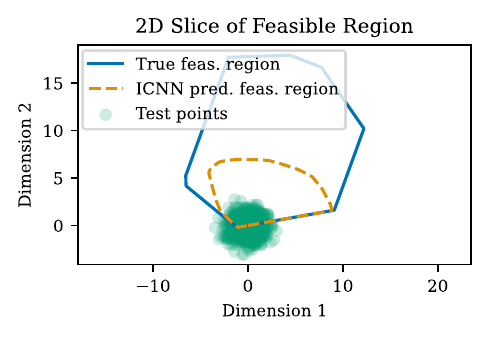}
    \caption{2-dimensional slice of the true feasible region and the predicted feasible region of a trained ICNN with hidden depth 1, with net injections from the test set overlaid.}
    \label{fig:slice}
\end{figure}

\subsection{Contingency Screening Results\label{section:experiments_screening}}

\begin{figure}
    \centering
    \includegraphics[width=\linewidth]{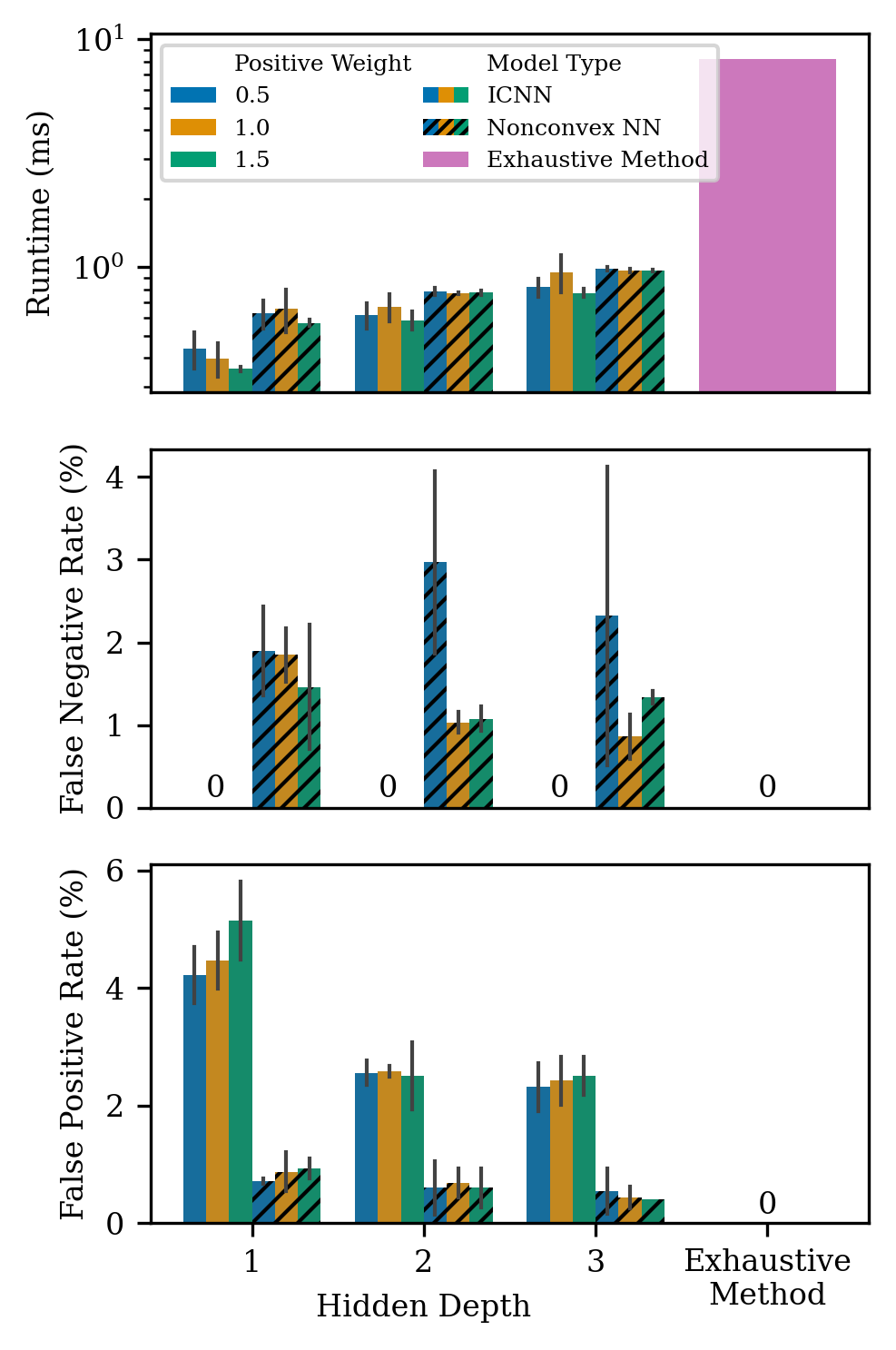}
    \caption{Results for our ICNN-based contingency analysis method, compared against a nonconvex neural network (NN) model and exhaustive checking of contingencies. (Top) Runtime to screen the feasibility of the 2,000 test injections. (Middle) False negative rate. (Bottom) False positive rate.}
    \label{fig:contingency_results}
\end{figure}

We show in Figure~\ref{fig:contingency_results} a comparison of the ICNNs trained via our methodology against standard NNs and the exhaustive method of checking all constraints individually for the contingency screening problem. Note that the ``Positive Weight'' value refers to the weight assigned to elements of the positive class in the training loss, where weights less than 1 typically encourage lower false positive rates, and weights greater than 1 typically encourage lower false negative rates.

Notably, the ICNNs trained with our differentiable scaling procedure in Algorithm~\ref{alg:training} achieve a speedup of $10\text{-}20\,\times$ over the exhaustive method, depending on the depth of the ICNN (Figure~\ref{fig:contingency_results}, top). Moreover, they uniformly achieve a false negative rate of 0 (Figure~\ref{fig:contingency_results}, middle), as guaranteed by our theoretical results, and a false positive rate between $2\%$ and $5\%$ (Figure~\ref{fig:contingency_results}, bottom). While the effect is not significant, it appears in the cases of hidden depth 1 and 3 that a lower positive weight may decrease the false positive rate of our approach. 

In comparison, the nonconvex NNs achieve a better false positive rate, ranging between 0.5\% and 1\%, but suffer significant false negative rates of $1\%$ to $3\%$, demonstrating that they cannot reliably be used for contingency screening, as they could misclassify infeasible scenarios as feasible. Our approach thus enables significantly faster screening than the exhaustive method while ensuring the reliability that cannot be guaranteed by standard neural networks.

\subsection{Faster Preventive Dispatch via SC-OPF \label{section:experiments_scopf}}

In practice, power system operators often want to perform \emph{preventive} dispatch to ensure that the chosen operating point will remain feasible in the case of contingencies. This problem, known as security-constrained (SC)-DC-OPF, adds to \eqref{opt:economic_dispatch} the additional constraint that $\xvec \coloneqq \pvec - \dvec$ should be feasible for all contingencies in the reference set $\calC$ -- that is, $\pvec - \dvec \in \calF_\calC$:
\begin{subequations} \label{opt:sc_opf}
\begin{align}
    \min_{\pvec \in \R^n} &\quad \sum_{i \in [n]} c_i(p_i) \\
    \mathrm{s.t.}&\quad \underline{\pvec} \leq \pvec \leq \overline{\pvec} \\
    &\quad \bv{1}^\top (\pvec - \dvec) = 0 \\
    &\quad \underline{\fvec} \leq \mat{H} (\pvec - \dvec) \leq \overline{\fvec}  \\
    &\quad \pvec - \dvec \in \calF_\calC \label{opt:security_constraint}
\end{align}
\end{subequations}

Because our ICNN approach to contingency screening yields an ICNN $f_\icnn(r^*\cdot)$ whose $0$-sublevel set is an inner approximation to $\calF_\calC$, one might naturally consider replacing the security constraint \eqref{opt:security_constraint} in the full SC-OPF problem with the conservative inner approximation $\hat{f}_\icnn\left(r^*(\pvec - \dvec)\right)\leq 0$ in an attempt to accelerate the solution time of this problem, since the original set $\calF_\calC$ is typically high-dimensional. We test the performance of this approach and its impact on system cost and infeasibility using our ICNN models trained on the IEEE 39-bus system, and we display the results in Figure~\ref{fig:scopf_results}.

\begin{figure}
    \centering
    \includegraphics[width=\linewidth]{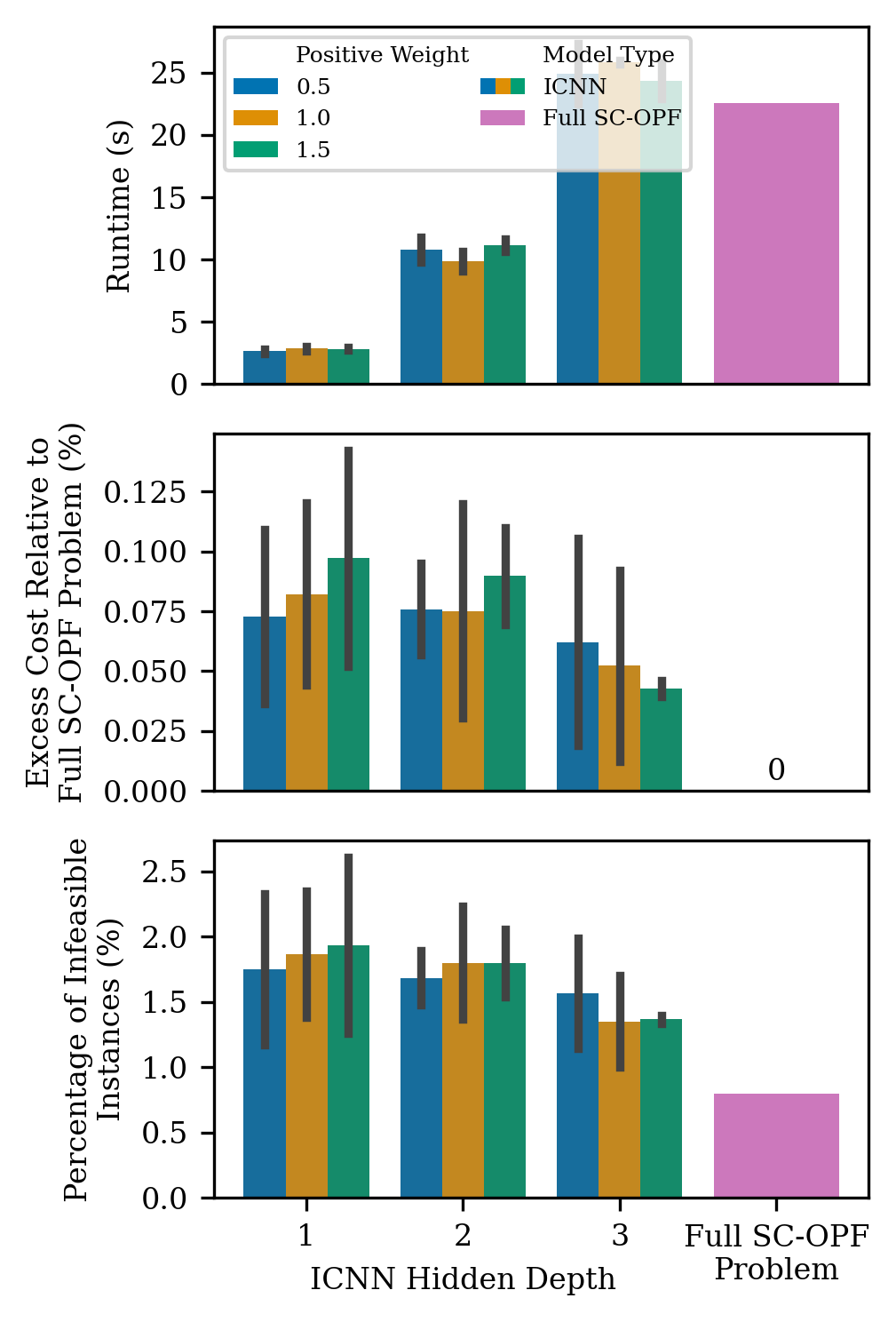}
    \caption{Results for the ICNN-based SC-OPF problem compared to the full SC-OPF problem \eqref{opt:sc_opf}. (Top) Runtime to solve the SC-OPF problem or ICNN version thereof on 2,000 test injections, disregarding infeasible injections. (Middle) Percent excess cost of the ICNN version of SC-OPF relative to the full SC-OPF problem \eqref{opt:sc_opf}. (Bottom) Percentage of infeasible demand instances for the ICNN version of SC-OPF compared against the full SC-OPF problem \eqref{opt:sc_opf}.}
    \label{fig:scopf_results}
\end{figure}

We see that, while the ICNNs with hidden depth 3 do not offer a speedup compared to solving \eqref{opt:sc_opf} exactly, the 2-layer ICNNs halve the runtime, and the shallowest 1-layer ICNNs speed up this problem by nearly a factor of 10 (Figure~\ref{fig:scopf_results}, top). Remarkably, they achieve this speedup while increasing the dispatch cost by no more than $0.1\%$ on average over the full SC-OPF problem (Figure~\ref{fig:scopf_results}, middle), and increasing the share of infeasible demand instances by only ${\sim}1\%$. It also appears that, for the ICNNs with hidden depth 1, decreasing the positive weight leads to better cost and less infeasibility. This agrees with intuition, since a lower positive weight encourages lower false positive rates, meaning that the ICNN should be a less conservative inner approximation to the set $\calF_\calC$. However, further study will be needed to determine whether this observation generalizes to deeper models, which in our experiments do not seem to exhibit this behavior. 

To conclude, note that we could modify our training methodology in Algorithm~\ref{alg:training} by replacing the classification loss with a differentiable convex optimization layer encoding the SC-OPF problem with ICNN security constraint. This would likely improve the performance of the ICNN for SC-OPF, since training the model end-to-end in such a manner would align training with the eventual downstream task faced by the model. We leave an implementation and evaluation of this change to future work.

\section{Discussion and Conclusions}

In this work, we proposed a methodology for data-driven training of input-convex neural network classifiers for contingency screening in power systems with zero false negative rate. We show that certifying and enforcing zero false negative rate -- i.e., reliability -- of an ICNN classifier can be achieved by solving a collection of optimization problems, and by incorporating these problems into a differentiable convex optimization layer during ICNN training, we can restrict training to be over the set of provably reliable models. We evaluate the performance of our approach on contingency screening and preventive dispatch on the IEEE 39-bus test system, showing that it achieves good performance, guaranteed reliability, and a significant computational speedup over conventional methods. We anticipate that the computational benefit of our approach will be even more significant for larger-scale power systems and higher-order contingency screening problems.

A number of interesting avenues remain open for future work, including (a) scaling up this approach to enable application to larger-scale power systems; (b) combining this screening approach with, e.g., methods from group testing to achieve comparable speedups for the full contingency analysis problem; and (c) extending this methodology to other applications that require constructing tractable inner approximations to some complicated set, such as learning data-driven and safe inner approximations to AC-OPF feasible regions or electric vehicle aggregate flexibility sets.

\begin{acks}
The authors thank Carey Priebe, Hayden Helm, and Christopher Yeh for illuminating discussions, and Kate Lytvynets for technical assistance. The authors acknowledge support from an NSF Graduate Research Fellowship (DGE-2139433), NSF Grant ECCS-1942326, the Resnick Sustainability Institute, a Caltech S2I Grant, and C3.ai Award \#11015.

\end{acks}

\bibliographystyle{ACM-Reference-Format}
\bibliography{main}

\end{document}

%% file: preamble.tex
\usepackage{amsmath, amsthm}
\usepackage{bm}
\usepackage{enumerate}
\usepackage[shortlabels]{enumitem}
\usepackage{mathtools}
\usepackage{graphicx}
\usepackage{titlesec}
\usepackage{braket}
\usepackage{float}

\usepackage{float}
\usepackage{mathrsfs}
\usepackage[ruled, lined, linesnumbered]{algorithm2e}
\SetKwInput{KwInput}{Input}

\usepackage{hyperref}
\usepackage{xcolor}
\hypersetup{
    colorlinks=true,
    linkcolor=blue,
    citecolor=magenta,
    filecolor=magenta,
    urlcolor=cyan,
}

\newcommand{\R}{{\mathbb R}}

\newcommand{\N}{{\mathbb N}}

\DeclareMathOperator{\vol}{vol}
\DeclareMathOperator{\inter}{int}

\DeclareMathOperator*{\argmax}{arg\,max}

\renewcommand{\vec}[1]{\mathbf{#1}}

\newcommand{\bv}[1]{\mathbf{#1}}

\newcommand{\calA}{{\mathcal{A}}}
\newcommand{\calB}{{\mathcal{B}}}
\newcommand{\calC}{{\mathcal{C}}}

\newcommand{\calF}{{\mathcal{F}}}

\newcommand{\calP}{{\mathcal{P}}}

\newcommand{\calX}{{\mathcal{X}}}

\newcommand{\avec}{{\vec{a}}}
\newcommand{\bvec}{{\vec{b}}}

\newcommand{\dvec}{{\vec{d}}}
\newcommand{\fvec}{{\vec{f}}}

\newcommand{\mvec}{{\vec{m}}}

\newcommand{\pvec}{{\vec{p}}}

\newcommand{\vvec}{{\vec{v}}}

\newcommand{\xvec}{{\vec{x}}}
\newcommand{\yvec}{{\vec{y}}}
\newcommand{\zvec}{{\vec{z}}}

\newcommand{\icnn}{{\mathrm{ICNN}}}
\newcommand{\relu}{{\mathrm{ReLU}}}

\newcommand{\mat}[1]{\mathbf{#1}}

\newcommand*\tageq{\refstepcounter{equation}\tag{\theequation}}

\newtheorem{assumption}{Assumption}
\newtheorem{definition}{Definition}
\newtheorem{problem}{Problem}
\newtheorem{theorem}{Theorem}

\newtheorem{proposition}{Proposition}

\makeatletter
\def\@proofnamefont{\scshape}
\def\@proofindent{\indent}
\ifcase\ACM@format@nr
\relax 
\or 
\or 
\or 
\or 
\or 
\or 
  \def\@proofnamefont{\itshape}
  \def\@proofindent{\noindent}
\or 
\or 
\or 
\or 
\fi